\documentclass[10pt]{article}
\usepackage[reset,a4paper,vmargin=2.5truecm,hmargin=2.0truecm]{geometry}

\usepackage{amsmath,amsthm,amssymb,bm}
\usepackage{amscd}
\usepackage[dvipdfmx]{graphicx}
\usepackage{colortbl}
\usepackage[all]{xy}
\usepackage{stmaryrd}
\usepackage{booktabs} 
\usepackage{framed, color}

\usepackage{amsmath,amsthm,amssymb,bm}

\newcommand{\bbC}{\mathbb{C}}

\newcommand{\bbR}{\mathbb{R}}

\newcommand{\fkg}{\mathfrak{g}}

\newcommand{\fkp}{\mathfrak{p}}

\newcommand{\cF}{\mathcal{F}} 
\newcommand{\cG}{\mathcal{G}}

\newcommand{\cK}{\mathcal{K}} 
\newcommand{\cL}{\mathcal{L}} 
 
\newcommand{\cN}{\mathcal{N}} 
 
\newcommand{\cP}{\mathcal{P}}

\newcommand{\bme}{\bm{e}}

\newcommand{\rM}{\mathrm{M}} 
 
\newcommand{\rO}{\mathrm{O}}

\newcommand{\rU}{\mathrm{U}}

\usepackage{amsmath,amsthm,amssymb,bm}

\DeclareMathOperator{\Ad}{Ad}

\DeclareMathOperator{\diag}{diag}

\DeclareMathOperator{\GL}{GL}

\DeclareMathOperator{\Herm}{Herm}

\DeclareMathOperator{\rank}{rank}

\DeclareMathOperator{\SO}{SO}

\DeclareMathOperator{\Sp}{Sp}
\DeclareMathOperator{\Sym}{Sym}

\DeclareMathOperator{\tr}{tr}

\newcommand{\gl}{\mathfrak{gl}}

\newcommand{\so}{\mathfrak{so}}
\newcommand{\fsp}{\mathfrak{sp}}

\newcommand{\T}{{}^t\!}

\newcommand{\mat}[1]{\begin{pmatrix}#1\end{pmatrix}}

\newcommand{\bmat}[1]{\begin{bmatrix}#1\end{bmatrix}}

\newcommand{\kakko}[1]{\left(#1\right)} 
\newcommand{\ckakko}[1]{\left\{#1\right\}} 

\newcommand{\D}{\displaystyle}

\newcommand{\acc}[1]{\textcolor{ocre}{#1}}

\numberwithin{equation}{section}

\theoremstyle{plain}
\newtheorem{theorem}{Theorem}[section]
\newtheorem{lemma}{Lemma}[section] 
\newtheorem{proposition}{Proposition}[section]
\newtheorem{cor}{Corollary}[section]

\theoremstyle{definition}

\theoremstyle{remark}

\twocolumn

\newcommand{\memo}[1]{#1}
\renewcommand{\memo}[1]{}	

\begin{document}

\title{Matrix-valued Bratu equation associated with the symmetric domains of type BDI and CI}
\author{
Hiroto Inoue\thanks{Institute of Mathematics for Industry, Kyushu University}
}
\date{\today}
\maketitle

\begin{abstract}
A formulation of the exponential matrix solution of the matrix-valued Bratu equation is given, based on the structure of the symmetric domain of type BDI. Moreover, an analog for the symmetric domain of type CI is given. 
\end{abstract}

\tableofcontents

\paragraph{Notation}
Throughout this article, 
$\Sym^+_{n}(\bbR)$ is the set of $n$-dimensional positive-symmetric matrices and $\mathrm{M}_{n,r}(\bbR)$ is the set of $(n\times r)$ real matrices.

\section{Introduction}

We consider the matrix-valued Bratu equation 
\begin{equation}\label{eq_mat_Bratu}
(h^{\prime}(s) h(s)^{-1})^{\prime}=(a\T a) h(s)^{-1}, \quad
a\in \mathrm{M}_{n,r}(\bbR)
\end{equation}
for a $\Sym^+_{n}(\bbR)$-valued function $h(s)$ of $s\in \bbR$ introduced in \cite{I2020}. 
The solution of the initial value problem of the Bratu equation \eqref{eq_mat_Bratu} was given in \cite{I2020} using a specific exponential matrix. 

In this article, we give a geometrical formulation of this solution based on the geometrical structure of the symmetric domains $G/K$. In detail, we express the exponential matrix solution of the equation \eqref{eq_mat_Bratu} by the defining function of the symmetric domain of type BDI according to Cartan's classification, ref. \cite{H, PS}. 
Moreover, we give an analog associated with the symmetric domain of type CI, as Table \ref{Tab}.

\begin{table}
\begin{center}
\caption{Matrix-valued differential equations associated to symmetric domains} \label{Tab}
$
\renewcommand{\arraystretch}{1.2} 
\begin{array}{c||c|c} 
\toprule
\text{domain}
&
\begin{array}{l}
\text{condition for}
\\
\text{$\Delta(w, e^{sB}w)$}
\end{array}
&
\begin{array}{l}
\text{equation of}
\\
\Delta(e^{sB}u_0, e^{sB}u_0)^{-1}
\end{array}
\\ \hline 
\text{BDI}
&\D (a\T a)\frac{s^2}{2}+O(s^3)
&
\begin{array}{r}
(h^{\prime} h^{-1})^{\prime}=(a\T a)h^{-1} \\ a\in \rM_{n,r}(\bbR) 
\end{array} 
\\ \hline 
\text{CI}
&\text{None}
&
\begin{array}{r}
(h'h^{-1})'=(ch^{-1})^2 \\ c\in \Sym_n(\bbR)
\end{array}
\\ \hline
\end{array}
\renewcommand{\arraystretch}{1} 
$
\end{center}
\end{table}

\section{Proof for exponential matrix solution}

\subsection{Block-Gauss decomposition}

\paragraph{Manifold $\Omega(J)$}
We define a submanifold $\Omega(J)$ in the cone $\Sym^+_{2n+r}(\bbR)$ by
\[
\Omega(J)=\ckakko{
G\in \Sym^+_{2n+r}(\bbR);\; JGJ=G^{-1}
}, 
\]
\[
J=\mat{0&0&I_n \\ 0&I_r&0 \\ I_n&0&0}. 
\]
Setting the set $V(J)$ by
\[
V(J)=\ckakko{B\in \Sym_{2n+r}(\bbR);\; JBJ=-B}, 
\]
we can consider the following map: 
\[
V(J)\ni B \mapsto \exp B \in \Omega(J). 
\]

We see that any element $B\in V(J)$ has the following form: 
\[
B=B(b,a,c):=\mat{b&a&c \\  \T a&0&-\T a \\ \T c&-a&-b}, \quad
\begin{array}{l}
b\in \Sym_n(\bbR), \\ a\in \mathrm{M}_{n,r}(\bbR), \\ c\in \mathrm{Alt}_n(\bbR). 
\end{array}
\]
For an element $B\in V(J)$, we define an $\Omega(J)$-valued function $G(s;B)$ by
\[
G(s; B) :=\exp(sB(b,a,c)) \quad (s\in \bbR). 
\]

We define the following sets of matrices: 
\begin{align*}
&\cG(J)=\ckakko{G\in \GL_{2n+r}(\bbR);\; JGJ=\T G^{-1}}\cong \mathrm{O}(n+r,n), 
\\
&\cN(J)
=\ckakko{N\in \cG(J);\; N=\mat{I_n&0&0 \\ *&I_r&0 \\ *&*&I_n}}, 
\\
&\qquad=\ckakko{N=\mat{I_n &0&0 \\ \T n_1&I_r&0 \\ \T n_2&-n_1&I_n},    n_2+\T n_2=-n_1 \T n_1}, 
\\
&\Omega_0(J)=\ckakko{A\in \Omega(J);\; A=\mat{*&0&0 \\ 0&*&0 \\ 0&0&*}}, 
\\
&\qquad=\ckakko{A=\mat{h&0&0 \\ 0&I_r&0 \\ 0&0&h^{-1}}, \T h=h}. 
\end{align*}

\begin{proposition}[Block-Gauss decomposition for $\Omega(J)$]
The following map is bijective: 
\[
\begin{array}{ccc}
\cN(J)\times \Omega_0(J) &\rightarrow& \Omega(J)\\
 (N, A)&\mapsto& NA\T N. 
 \end{array}
\]
\end{proposition}

We denote the variable change defined through this decomposition by
\[
G \rightarrow N, A \rightarrow h, n_1, n_2, 
\]
or simply by $G  \rightarrow h$.

Now we state the exponential matrix solution in [I2020] as follows.

\begin{proposition}\label{prop_3}
For $B=B(b,a,0)\in V(J)$, define a $\Sym^+_n(\bbR)$-valued function $\tilde{h}(s)$ by the variable change
\[
G(s;B)\rightarrow \tilde{h}(s). 
\]
Then, $\tilde{h}(s)$ satisfies the matrix-valued Bratu equation; 
\[
(\tilde{h}^{\prime}(s) \tilde{h}(s)^{-1})^{\prime}=(a\T a) \tilde{h}(s)^{-1}. 
\]
\end{proposition}

In the rest of this section, we give an alternative proof for this proposition based on the Lagrangian calculus.


\subsection{Lagrangian calculus}


\begin{lemma}\label{lem_1}
Let $G=G(s)$ be an $\Omega(J)$-valued function. Under the variable change
\[
G \rightarrow N, A \rightarrow h, n_1, n_2, 
\]
the following equation holds: 
\begin{align*}
&\frac12\tr((G' G^{-1})^2)
\\
&=\frac12\tr(( A' A^{-1})^2)+\tr(A\T N' \T N^{-1} A^{-1}N^{-1} N')
\\
&=\tr(( h' h^{-1})^2)+2\tr(hn_1' \T n_1')+\tr(h\T mhm), 
\end{align*}
\[
m=n_1 \T n'_1+n'_2=-\T m. 
\]
\end{lemma}

\begin{proof}
To prove this lemma, we use the following orthogonal relation.

\begin{lemma}\label{lem_2}
Define a subset of $M=\mathrm{M}_{2n+r}(\bbR)$ as
\[
M_{\leq 0}=\ckakko{X\in M;\; X=\mat{*&0&0 \\ *&*&0 \\ *&*&*}}, 
\]
\[
M_{< 0}=\ckakko{X\in M;\; X=\mat{0&0&0 \\ *&0&0 \\ *&*&0}}. 
\]
Then, it holds that
\[
\tr(X_1 X_2)=0
\]
for any $X_1\in M_{\leq 0}, X_2 \in M_{<0}$. 
\end{lemma}

Inserting $G=NA\T N$ into $G'G^{-1}$, we get
\begin{align*}
&G'G^{-1}
\\
&= NA \T N' \T N^{-1} A^{-1} N^{-1} +  N A' A^{-1}N^{-1} + N' N^{-1}
\\
&=:X_1 + X_2 + X_3. 
\end{align*}
Here we see that
\[
X_1\in \Ad(N) \T M_{<0}, \quad X_2\in \Ad(N) \T M_{\leq 0} \cap M_{\leq 0},  
\]
\[
X_3 \in M_{<0}. 
\]
Then, by Lemma \ref{lem_2}, we get
\begin{align*}
\frac12\tr((G' G^{-1})^2)&=\tr(X_1 X_3)+\frac12\tr(X_2 X_2). 
\end{align*}

Inserting
\begin{align*}
A'A^{-1}&=\mat{h' h^{-1}&0&0 \\ 0&0&0 \\ 0&0&-h^{-1}h'}, 
\\
N^{-1}N'
&=\mat{0&0&0 \\ \T n'_1&0&0 \\ m&-n'_1&0}, \; m=n_1 \T n'_1+n'_2, 
\end{align*}
we calculate as
\begin{align*}
\tr(X_1 X_3)&=\tr(A\T N' \T N^{-1} A^{-1}N^{-1} N')
\\
&=2\tr(hn_1' \T n_1')+\tr(h\T mhm),
\\
\frac12\tr(X_2 X_2)&=\frac12\tr(( A' A^{-1})^2)=\tr(( h' h^{-1})^2). 
\end{align*}
Therefore, we get the Lagrangian as
\begin{align*}
\frac12\tr((G' G^{-1})^2)=&\tr(( h' h^{-1})^2)
+2\tr(hn_1' \T n_1')
\\
&+\tr(h\T mhm). 
\end{align*}
\end{proof}

Next we proceed to derive the Euler-Lagrange equation. For the preparation, we deal with the constrain condition for $G(s)$ as follows.  
\begin{lemma}
For the Lagrangian
\[
\cL=\frac12\tr((G'G^{-1})^2)
+\tr((JGJ\T G-I)\lambda)+\tr((G-\T G)\mu)
\]
for $\GL(n,\bbR)$-valued function $G, \lambda, \mu $, 
the Euler-Lagrange equation reads
\begin{align*}
&
\begin{cases}
(G'G^{-1})'=0, 
\\
G(\mu-\T \mu)G+J(\mu-\T\mu)J=0. 
\end{cases}
\end{align*}
\end{lemma}
\begin{proof}
By the property of the matrix trace, we have the equation
\begin{align*}
\tr(JG_1J\T G_2\lambda)
&=\tr(G_1 \cdot J\T G_2 \lambda J)
\\
&=\tr(G_2 \cdot J\T G_1 J\T\lambda ). 
\end{align*}
Then, we can derive the Euler-Lagrange equation as
\begin{eqnarray*}
G:&&G^{-1}(G''-G'G^{-1}G')G^{-1}=J\T G\lambda J+J\T GJ\T \lambda, 
\\
&&\Leftrightarrow
(G'G^{-1})'=J\lambda J+\T \lambda+G(\mu-\T \mu), 
\\
\lambda: &&GJ\T G=J, 
\\
\mu: &&G-\T G=0. 
\end{eqnarray*}
Applying the map
$X\mapsto X-J\T XJ$ to the both sides of the first equation, we get
\begin{align*}
&2(G'G^{-1})'=G(\mu-\T \mu)+J(\mu-\T\mu)JG^{-1}
\\
\Leftrightarrow&
2(G'G^{-1})'G=G(\mu-\T \mu)G+J(\mu-\T\mu)J. 
\end{align*}
We see that the left hand side is a symmetric matrix and the right hand side is a skew-symmetric matrix. Therefore, they turn to be $0$. 
\end{proof}

\paragraph{Derivation of Euler-Lagrange equation}
Now we derive the Euler-Lagrange equation for our Lagrangian
\begin{align*}
\cL&=\cL(h,n_1,n_2,h',n_1',n_2')
\\
&=\tr(( h' h^{-1})^2)
+2\tr(hn_1' \T n_1')+\tr(h\T mhm), 
\\
&m=n_1 \T n'_1+n'_2. 
\end{align*}
The result is written as follows:  
\begin{eqnarray*}
\text{variable } x: &\text{E-L eq. }\D \frac{d}{ds}\frac{\partial \cL}{\partial x'}=\frac{\partial \cL}{\partial x} \notag
\\ \hline  \\
h: &
(h^{\prime} h^{-1})^{\prime}=hn_1'\T n_1'+ h\T mhm, 
\\
n_1: &
\ckakko{2hn_1'+h\T mhn_1}' =h mhn_1', 
\\
n_2: &
(h\T mh)'=0. 
\end{eqnarray*}

From the third equation, we get
\[
h\T mh=\tilde{c}  : \text{constant matrix}, \quad \tilde{c}\in \mathrm{Alt}_{n}(\bbR). 
\]

Thus, the above equations become
\begin{eqnarray*}
h: &(h^{\prime} h^{-1})^{\prime}=hn_1'\T n_1'-  \tilde{c} h^{-1}\tilde{c}h^{-1},
\\
n_1: &2(hn_1')'+2(\tilde{c}n_1)' =0. 
\end{eqnarray*}

From the second equation, we get
\[
hn_1' +{c}n_1 =\tilde{a}: \text{constant matrix}. 
\]

Especially, in the case $\tilde{c}=0$, we get
$
hn_1' =\tilde{a}
$.  
Inserting this into the first equation, we obtain the matrix-valued Bratu equation
\begin{eqnarray*}
h: &(h^{\prime} h^{-1})^{\prime}=(\tilde{a}\T \tilde{a}) h^{-1}. 
\end{eqnarray*}

Next we identify the constant matrices $\tilde{c}, \tilde{a}$. 

Defining the function $G(s)$ by the relation $G\rightarrow h, n_1, n_2$, we have that
\[
G'G^{-1}=\mat{*&hn_1'-hmhn_1&h\T mh \\ *&*&* \\ *&*&*}
=\mat{*&\tilde{a}&\tilde{c} \\ *&*&* \\ *&*&*}. 
\]
On the other hand, $G(s)$ satisfies the original E-L equation
\[
G'G^{-1}=\text{(constant matrix)}. 
\]
This shows that $G(s)$ is given as
\[
G=G(s; B(\tilde{b},\tilde{a}, \tilde{c})), \quad \tilde{b}\in \Sym_n(\bbR). 
\]

We summarize the above calculation in the following table: 
{\small
\[
\begin{array}{ccc}\hline \\
G(s;B(b,a,c))&\rightarrow&h(s), n_1(s), n_2(s)
\\ \\
(G'G^{-1})'=0&&
\begin{cases}
(h^{\prime} h^{-1})^{\prime}=hn_1'\T n_1'-  {c} h^{-1}{c}h^{-1} 
\\
hn_1' +{c}n_1={a}
\\
h\T(n_1 \T n'_1+n'_2)h=c 
\end{cases}
\\ && \\
\D\frac12\tr((G'G^{-1})^2)&
&\tr(( h' h^{-1})^2)
+2\tr(hn_1' \T n_1')+\tr(h\T mhm)
\\ \\ \hline
\end{array}
\]
}


\begin{theorem}
The equation
\[
G'(s)G(s)^{-1}=B(b,a,c)
\]
for an $\Omega(J)$-valued function $G(s)$ is expressed as
\begin{equation}\label{eq_hn1n2}
\begin{cases}
(h^{\prime} h^{-1})^{\prime}=hn_1'\T n_1'-  {c} h^{-1}{c}h^{-1}, 
\\
hn_1' +{c}n_1={a}, 
\\
h\T(n_1 \T n'_1+n'_2)h=c 
\end{cases}
\end{equation}
by the variable change $G(s)\rightarrow h(s), n_1(s), n_2(s)$. 
Especially, in the case ${c}=0$, it is equivalent to the matrix-valued Bratu equation
\[
(h^{\prime} h^{-1})^{\prime}=({a}\T {a}) h^{-1}. 
\]
\end{theorem}

The exponential matrix solution (Proposition \ref{prop_3}) follows from this theorem.

\memo
{

\begin{corollary}
方程式\eqref{eq_hn1n2}は, 変数変換
\begin{align*}
H(s):=\mat{I_n &0 \\ \T n_1(s) & I_r}\mat{h(s) & 0 \\ 0&I_r}\mat{I_n & n_1(s) \\ 0&I_r}
\end{align*}
によって
\begin{align*} 
\mat{I_n&0}(H'H^{-1})'=-c\mat{\T n_2&n_1}'
\end{align*}
と表される. 
\end{corollary}
\begin{proof}
定義より$H^{-1}$は
\begin{align*}
H^{-1}&=\mat{I_n & -n_1 \\ 0 & I_r}\mat{h^{-1}& 0 \\ 0 &I_r}\mat{I_n&0 \\ -\T n_1 & I_r}
\\
&=\mat{h^{-1}+n_1 \T n_1 & -n_1 \\ -\T n_1 & I_r}
\end{align*}
なので, $H'H^{-1}$は
\begin{align*}
&H'H^{-1}=
\\
&\mat{h'h^{-1}+h'n_1\T n_1-(hn_1)'\T n_1&-h' n_1+(hn_1)' \\ (\T n_1 h)'(h^{-1}+n_1 \T n_1)-(\T n_1 hn_1)' \T n_1&-(\T n_1 h)'n_1 +(\T n_1 hn_1)'}
\\
&=\mat{h'h^{-1}-hn_1'\T n_1&hn_1' \\ 
(\T n_1 h)'h^{-1}-\T n_1 hn_1' \T n_1&\T n_1 hn_1'}
\end{align*}
と計算される. よって方程式\eqref{eq_hn1n2}
は次のように書ける: 
\begin{align*} 
\mat{I_n&0}(H'H^{-1})'&=\mat{-(hn_1')' \T n_1-ch^{-1}ch^{-1}&-cn_1'}
\\
&=\mat{cn_1' \T n_1-ch^{-1}ch^{-1}&-cn_1'}
\\
&=-c\mat{\T n_2&n_1}'. 
\end{align*}
\end{proof}

以上の議論を次のように整理する: 

\[
\begin{array}{ccccc}\hline \\
G(s;B(b,a,c))&\rightarrow&H(s),n_2(s)&\rightarrow
\\ \\
(G'G^{-1})'=0&&(I_n\; 0)(H'H^{-1})'=0
\\ && \\
\D\frac12\tr((G'G^{-1})^2)&&\tr((H'H^{-1})^2)+\tr(H\Lambda)&
\\ \\ \hline
\end{array}
\]
\[
\begin{array}{ccccc}\hline \\
h(s), n_1(s), n_2(s)
\\ \\
\begin{cases}
(h^{\prime} h^{-1})^{\prime}=hn_1'\T n_1'-  {c} h^{-1}{c}h^{-1}
\\
hn_1'-{c}n_1 ={a}
\end{cases}
\\ && \\
\tr(( h' h^{-1})^2)
+2\tr(hn_1' \T n_1')+\tr(h\T mhm)
\\ \\ \hline
\end{array}
\]

}

\memo
{

$\cN(J)\times \Omega_0(J)$の測地線の方程式
\begin{align*}
\cL&=\frac12\tr((G' G^{-1})^2)
\\
&=\tr(( h' h^{-1})^2)+2\tr(hn_1' \T n_1')+\tr(h\T mhm)
\end{align*}
\[
m=n_1 \T n'_1+n'_2. 
\]
共役座標: 
\begin{align*}
&\pi:=\frac{\partial \cL}{\partial h'}=2h^{-1}h' h^{-1}
&&\Leftrightarrow h'=\frac12h\pi h
\\
&\nu_1:=\frac{\partial \cL}{\partial n_1'}=4hn_1'+2h\T m h n_1
&&\Leftrightarrow n_1'=\frac14 h^{-1}\nu_1 - \frac12 \T mh n_1
\\
&\nu_2:=\frac{\partial \cL}{\partial n_2'}=2hmh
&&\Leftrightarrow n_2'=m-\frac14 n_1 \T \nu_1 h^{-1}
\end{align*}

ラグランジアンの各項は
\begin{align*}
&\tr((h' h^{-1})^2)=\frac14 \tr((h\pi)^2)
\\
&2\tr(hn_1' \cdot \T n_1')=\tr\kakko{h(\frac12 h^{-1}\nu_1- \T m h n_1)(\frac12 \T\nu_1h^{-1}- \T n_1 h m)}
\\
&=\frac18 \tr(\nu_1 \T \nu_1 h^{-1})-\frac12 \tr(\T m hn_1 \T \nu_1) 
+\frac12 \tr(\T m h n_1 \T n_1 hm)
\\
&=\frac18 \tr(\nu_1 \T \nu_1 h^{-1}) 
-\frac14 \tr(h^{-1}\T \nu_2 n_1 \T \nu_1)
\\
&+\frac18\tr(h^{-1}\T\nu_2 n_1 \T n_1 \nu_2 h^{-1})
\\
&\tr(h\T mhm)=\frac14\tr\kakko{\T \nu_2 h^{-1}\nu_2 h^{-1}}
\end{align*}
Legendre変換: 
\begin{align*}
&H=\tr(h'\pi)+\tr(n_1' \T\nu_1)+\tr(n_2'\T\nu_2)-L
\\
&=\frac12\tr((h\pi)^2)+\frac14 \tr(h^{-1}\nu_1 \T\nu_1)-\frac14\tr(h^{-1}\T\nu_2 n_1 \T \nu_1)
\\
&+\frac12 \tr(h^{-1}\nu_2 h^{-1}\T\nu_2)-\frac14\tr(n_1\T\nu_1 h^{-1}\T\nu_2)
\\
&-\frac14 \tr((h\pi)^2)-\frac18\tr(\nu_1\T\nu_1 h^{-1})+\frac14\tr(h^{-1}\T\nu_2 n_1 \T \nu_1)
\\
&-\frac18\tr(h^{-1}\T\nu_2 n_1\T n_1 \nu_2 h^{-1})-\frac14 \tr(\T\nu_2 h^{-1}\nu_2 h^{-1})
\\
&=\frac14 \tr((h\pi)^2)+\frac18\tr(\nu_1\T\nu_1 h^{-1})
+\frac14 \tr(h^{-1}\nu_2 h^{-1}\T\nu_2)
\\
&-\frac14\tr(n_1\T\nu_1 h^{-1}\T\nu_2)+\frac14\tr(h^{-1}\T\nu_2 n_2 \nu_2 h^{-1})
\end{align*}

$\Sym^+_n(\bbR)\times \bbR^{n\times r}$の測地線の方程式 \\
(共役座標を積分定数で置き換えられる場合)
\[
\cL=\T \mu' \sigma^{-1} \mu' +\frac12 \tr((\sigma'\sigma^{-1})^2)
\]
共役座標: 
\begin{align*}
&\nu:=\frac{\partial \cL}{\partial \mu'}=2\sigma^{-1}\mu'
\\
&\pi:=\frac{\partial \cL}{\partial \sigma'}=\sigma^{-1}\sigma'\sigma^{-1}
\\
&H(\mu, \sigma, \nu, \pi)=L=\frac14 \T \nu\sigma\nu+\frac12\tr((\pi\sigma)^2)
\end{align*}
正準方程式は
\begin{eqnarray*}
&&\mu'=\frac{\partial H}{\partial \nu}=\frac12 \sigma\nu
\\
&&\nu'=-\frac{\partial H}{\partial \mu}=0
\\
&&\sigma'=\frac{\partial H}{\partial \pi}=\sigma\pi\sigma
\\
&&\pi'=-\frac{\partial H}{\partial \sigma}=-\pi \sigma\pi-\frac14 \nu\T\nu
\end{eqnarray*}
これより
\[
(\sigma'\sigma^{-1})'=-a\T a \sigma^{-1}. 
\]
ただし
\[
\nu=2\sigma^{-1}\mu'=:a\; (定数). 
\]
ラグランジアンになおすと, 
\begin{align*}
L(\sigma)&=\tr\kakko{\pi \frac{\partial H}{\partial \pi}}+\tr\kakko{\nu \frac{\partial H}{\partial \nu}}-H
\\
&=\frac12 \tr((\sigma'\sigma^{-1})^2)-\T a\sigma a
\end{align*}

}

\clearpage

\section{Expression in symmetric domains}

Considering a realization of $\Omega(J)$ as a symmetric domain, 
we express
\begin{itemize}
\item
the condition $c=0$ for $B=B(b,a,c)$, 
\item
the variable change $G\rightarrow h$ by the block-Gauss decomposition
\end{itemize}
in Proposition \ref{prop_3} in terms of the structure of symmetric domains.

\subsection{Siegel domain of type BDI}

\paragraph{Siegel domain $D(J)$ of type BDI}
For the set
\begin{align*}
&M=\ckakko{U=\mat{U_1 \\ U_2 \\ U_3}\in \rM_{2n+r,n}(\bbR);\; \rank U=n}, 
\end{align*}
we consider the group actions
\[
\GL(2n+r,\bbR) \curvearrowright M \curvearrowleft \GL(n,\bbR)
\]
by the matrix multiplication. For the quatient
\[
D:=M/\GL(n,\bbR)=\ckakko{u=[U]=\bmat{U_1 \\ U_2 \\ U_3};\; U\in M}, 
\]
we equip it with the induced topology to consider $D$ as a domain. 
We let
\[
J=-\mat{0&0&I_n \\ 0&I_r&0 \\  I_n&0&0}
\]
and define a subdomain $D(J)\subset D$ by
\[
M(J)=\ckakko{U\in M;\; \T UJU>0}, 
\]
\[
D(J):=\ckakko{u=[U]\in D;\; U\in M(J)}. 
\]

\begin{lemma}[Realization as a Siegel domain]
The domain $D(J)$ is the following set: 
\[
D(J)= \ckakko{u=\bmat{U_1 \\ U_2 \\ I_n}\in D;\; 
-\T U_1- U_1-\T U_2 U_2>0 }. 
\]
\end{lemma}

We define a subdomain $\Sigma \subset D$ (the Shilov boundary of $D(J)$) by
\[
\Sigma= \ckakko{u=\bmat{U_1 \\ U_2 \\ I_n}\in D;\; 
-\T U_1- U_1-\T U_2 U_2=0 }. 
\]

We take two points $u_0, w\in D$ as
\begin{align*}
&u_0=[U_0]\in D(J), \quad w=[W]\in \Sigma, 
\\
&U_0=\mat{-I_n \\ 0 \\ I_n}, \quad W=\mat{0 \\ 0 \\ I_n}. 
\end{align*}

For two points $u_1, u_2 \in D(J)$ expressed as
\[
u_j=[U_j], \; U_j =\mat{U_{j,1} \\ U_{j,2} \\ U_{j,3}} \quad (j=1,2), 
\]
we define $\Delta(u_1,u_2) \in \rM_n(\bbR)$ by
\begin{align*}
&\Delta(u_1,u_2):=\T (U_1 \Gamma_1^{-1})J(U_2\Gamma_2^{-1}), 
\end{align*}
\[
\Gamma_j=\T WU_j=U_{j,3}\quad (j=1,2). 
\]

\paragraph{Expression of $D(J)$ as a symmetric space $G/K$}
Consider the action $\cG(J) \curvearrowright D(J)$ defined by
\[
g\cdot u=[gU] \quad (g\in \cG(J), u=[U]\in D(J)). 
\]
We also define the following subgroups: 
\begin{align*}
&\cP(J)=\ckakko{g\in \cG(J);\; g=\mat{*&*&* \\ 0&*&* \\ 0&0&*}}, 
\\
&\cK(J)=\ckakko{g\in \cG(J);\; \T g=g^{-1}}\cong\rO(n+r)\times \rO(n), 
\\
&\cG(J)_{u_0}=\ckakko{g\in \cG(J);\; gu_0=u_0}. 
\end{align*}

\begin{lemma}
The following assertions hold: 
\begin{enumerate}
\item
$\cG(J)=\cP(J)\cK(J)$. 
\item
$\cG(J)_{u_0}=\cK(J)$. 
\item
The action $\cG(J)\curvearrowright D(J)$ is transitive, and it holds that
\[
D(J)
\cong \SO_0(n+r,n)/\SO(n+r)\times \SO(n). 
\]
\end{enumerate}
\end{lemma}
\begin{proof}
{\it 2}. 
Let
\[
A=\frac{1}{\sqrt{2}}\mat{I_n&0&-I_n \\ 0&\sqrt{2}I_r&0 \\ I_n&0&I_n}, 
\]
and
\[
v_0:=Au_0=\bmat{ I_n \\ 0 \\ 0}, \quad
K:=AJ\T A=\mat{I_n&0 \\ 0&-I_{n+r}},
\]
\[
\cG(K)= A\cG(J) \T A. 
\]
For $g\in \GL(2n+r,\bbR)$, we see the equivalence
\begin{align*}
gv_0=v_0 \Leftrightarrow g=\mat{* &*\\ 0& *}. 
\end{align*}
Thus, for $g\in \cG(K)$, it holds the equivalence
\begin{align*}
gv_0=v_0 \Leftrightarrow  g=\T g^{-1}. 
\end{align*}
Therefore, for $g\in \cG(J)$, 
\begin{align*}
gu_0=u_0 \Leftrightarrow  g=\T g^{-1}. 
\end{align*}

\end{proof}

\paragraph{Isomorphism $\cF: \Omega(J)\xrightarrow{\sim} D(J)$}
We define a map $\cF: \Omega(J)\xrightarrow{} D$ by
\[
\cF(G)=\bmat{ G_{3}-I_n \\ G_{2} \\ h}, \quad
G=\mat{h &*&* \\ G_2&*&* \\ G_3&*&*}\in \Omega(J), 
\]
and a map $\pi: \Omega(J)\rightarrow\Sym^+_n(\bbR)$ by
\[
\pi(G)=h, \quad
G=\mat{h &*&* \\ G_2&*&* \\ G_3&*&*}\in \Omega(J). 
\]
We define an action $\cG(J)\curvearrowright \Omega(J)$ by
\[
g\cdot G=\T g^{-1}Gg^{-1}\quad (g\in \cG(J), G\in \Omega(J)). 
\]

\begin{lemma}
\begin{enumerate}
\item
For any $g\in \cG(J), G\in \Omega(J)$, it holds that
\[
\cF(g\cdot G)=g\cF(G). 
\]
\item
It holds that $\cF(\Omega(J))\subset D(J)$, and $\cF: \Omega(J)
\rightarrow D(J)$ is bijective. 
\item
For any $G\in \Omega(J)$, it holds that
\[
\frac12\pi(G)=\Delta(\cF(G),\cF(G))^{-1}. 
\]
\end{enumerate}
\end{lemma}
By this lemma, we obtain the following commutative diagram: 
\[
\begin{xy}
\xymatrix{
\Omega(J) \ar[r]^-{\cF} \ar[dr]^{\circlearrowright}_-{\frac12\pi}& D(J) \ar[d]^-{\Delta(\cdot,\cdot)^{-1}}
\\
&\Sym^+_n(\bbR)
}
\end{xy}
\]
\begin{proof}
{\it 1}. 
First, we prove for the case $G=I_{2n+r}$. 
Take an arbitrary $g\in \cG(J)$ and express it as
\[
g=pk, \quad p\in \cP(J), \; k\in \cK(J)
\]
\[
p=\mat{p_1&-p_2&\T p_3 \\ 0&I_r&\T p_2 \\ 0&0&\T p_1^{-1}}. 
\]
Then, from the relation $\T p^{-1}=JpJ$ we get
\begin{align*}
p\cdot I_{2n+r}&=\T p^{-1}p^{-1}
\\
&=\mat{\T p_1^{-1}&0&0 \\ \T p_2&I_r&0 \\ \T p_3&-p_2& p_1}
\mat{ p_1^{-1}&p_2& p_3 \\ 0&I_r&-\T p_2 \\ 0&0& \T p_1}
\\
&=\mat{
\T p_1^{-1}p_1^{-1} &*&* \\ 
\T p_2  p_1^{-1}&*&* \\ 
\T p_3 p_1^{-1}&*&*}. 
\end{align*}
Therefore, its image under the map $\cF$ is
\begin{align*}
\cF(p\cdot I_{2n+r})&=\bmat{\T p_3 p_1^{-1}-I_n\\
\T p_2  p_1^{-1} \\
\T p_1^{-1}p_1^{-1}}
=\bmat{\T p_3 -p_1 \\
\T p_2    \\
\T p_1^{-1}}. 
\end{align*}
On the other hand, we see that
\begin{align*}
p\cF(I_{2n+r})=p\bmat{-I_n \\ 0 \\ I_n}
=\bmat{-p_1 +\T p_3  \\
\T p_2    \\
\T p_1^{-1}}. 
\end{align*}
It shows that $\cF(p\cdot I_{2n+r})=p\cF(I_{2n+r})$. Therefore, 
\begin{align*}
&\cF(g\cdot I_{2n+r})=\cF(p\cdot I_{2n+r})
\\
&=pu_0=pku_0=g\cF(I_{2n+r}). 
\end{align*}

Next, we show the assertion for the general case $g\in \cG(J), G\in \Omega(J)$. 
We express $G$ as
\[
G=g_0\cdot I_{2n+r} , \quad g_0\in \cG(J). 
\]
Then, we see that
\begin{align*}
&\cF(g\cdot G)=\cF(gg_0\cdot I_{2n+r})
=gg_0\cF( I_{2n+r})
\\
&=g\cF(g_0 \cdot I_{2n+r})
=g\cF(G). 
\end{align*}

{\it 2}. 
The equation $\cF(\Omega(J))=D(J)$ is follows
from {\it 1} and the transitivity of the action of $\cG(J)$. 
The injectivity of $\cF$ holds because
the isotropy subgroups are coincident to be $\cK(J)$.

{\it 3}. 
Take an arbitrary $G\in \Omega(J)$ and express it as
\[
G=\mat{h &*&* \\ G_2&*&* \\ G_3&*&*}, \quad 
hG_3 +\T G_3 h + \T G_2 G_2=0. 
\]
Thus, its image is
\begin{align*}
\cF(G)=[U], \quad 
U=\mat{U_1 \\ U_2 \\ U_3}=\mat{ G_{3}-I_n \\ G_{2} \\ h}. 
\end{align*}
Setting $\Gamma:=h$, we get 
\begin{align*}
&\Delta(\cF(G), \cF(G))=\T (U\Gamma^{-1})J(U\Gamma^{-1})
\\
&=-\T(G_3 h^{-1}-h^{-1})-(G_3 h^{-1}-h^{-1})-h^{-1} \T G_2 G_2 h^{-1} 
\\
&=2h^{-1} =2\pi(G)^{-1}. 
\end{align*}
\end{proof}

\paragraph{Characterization of $B=B(b,a,0)$}
\begin{align*}
&\fkg(J)=\ckakko{X\in \gl(2n+r,\bbR);\; XJ +J\T X=0}\cong \so(n+r,n), 
\\
&\fkp(J)=\ckakko{X\in \fkg(J);\; X=\T X}. 
\end{align*}
An arbitrary element $\fkp(J)$ has the form
\[
B=B(b,a,c):=\mat{b&a&c \\  \T a&0&-\T a \\ \T c&-a&-b}, \quad
\begin{array}{l}
b\in \Sym_n(\bbR), \\ a\in \mathrm{M}_{n,r}(\bbR), \\ c\in \mathrm{Alt}_n(\bbR)
\end{array}. 
\]
Let $w\in \Sigma$ as above, 

\begin{lemma}
For an element $B(b,a,c)\in \fkp(J)$, it holds the following Taylor expansion: 
\[
\Delta(w, e^{sB}w)=-cs+\kakko{a\T a-bc-cb}\frac{s^2}{2}+O(s^3). 
\]
\end{lemma}
\begin{proof}
Let
\[
\Gamma_1=\T W W=I_n, \quad \Gamma(s)=\T We^{sB}W, 
\]
\[ 
\Phi(s)=\T WJe^{sB}W. 
\]
Then, we have
\begin{align*}
\Delta(w, e^{sB}w)&=\T (W\Gamma_1^{-1})J(e^{sB}W\Gamma(s)^{-1})
\\
&=\Phi(s)\Gamma(s)^{-1}. 
\end{align*}
By the definition, it is easily seen that
\begin{align*}
&\Phi(s)=-cs+ (a\T a-bc+cb) \frac{s^2}{2}+O(s^3), 
\\
&\Gamma(s)^{-1}=I_n+bs+O(s^2). 
\end{align*}
Then, we obtain the expansion 
\[
\Delta(w, e^{sB}w)=-cs+\kakko{a\T a-bc-cb}\frac{s^2}{2}+O(s^3). 
\]
\end{proof}

\begin{proposition}
For $B\in \fkp(J)$, assume the Taylor expansion 
\[
\Delta(w, e^{sB}w)=(a\T a) \frac{s^2}{2} +O(s^3), \quad a\in \rM_{n,r}(\bbR)
\]
at $s=0$. 
Define a $\Sym^+_n(\bbR)$-valued function $h(s)$ by
\[
h(s)=\Delta(e^{sB}u_0, e^{sB}u_0)^{-1}. 
\]
Then, $h(s)$ satisfies the matrix-valued Bratu equation; 
\[
(h'h^{-1})'=2(a\T a) h^{-1}. 
\]
\end{proposition}

\begin{proof}
From the assumption, the element $B$ has the form $B=B(b,a,0)$. 
Since
\[
\cF^{-1}(e^{sB}u_0)=e^{-2sB}=:G(s), 
\]
the function $\tilde{h}(s):=\pi(G(s))$ satisfies the matrix-valued Bratu equation
\[
(\tilde{h}'\tilde{h}^{-1})'
=4(a\T a)\tilde{h}^{-1}. 
\]
Therefore, the function
\[
{h}(s):=\Delta(e^{sB}u_0, e^{sB}u_0)^{-1}=\frac12 \tilde{h}(s)
\]
satisfies the equation
\[
({h}'{h}^{-1})'=2(a\T a) {h}^{-1}. 
\]
\end{proof}

%

\newpage
\subsection{Bounded domain of type BDI and power series solution}

\paragraph{Bounded domain $D(L)$ of type BDI}
Let
\begin{align*}
L=\mat{I_n&0&0 \\ 0&-I_r&0 \\  0&0&-I_n}. 
\end{align*}
We define a subdomain $D(L)\subset D$ by
\begin{align*}
&M(L)=\ckakko{U\in M;\; \T ULU>0}, 
\\
&D(L):=\ckakko{u=[U]\in D;\; U\in M(L)}. 
\end{align*}

We notice the relation $L=AJ\T A$ with a matrix
\[
A=\frac{1}{\sqrt{2}}\mat{I_n&0&-I_n \\ 0&\sqrt{2}I_r&0 \\ I_n&0&I_n}. 
\]

\begin{lemma}
\begin{enumerate}
\item
It holds the following bijection: 
\begin{align*}
&D(J)\xrightarrow{\sim} D(L) : u\mapsto Au
\end{align*}
\item
(Realization as a bounded domain) 
$D(L)$ is described as
\[
\hspace{-0.6cm}
D(L)= \ckakko{v=\bmat{I_n \\ V_2 \\ V_3}\in D;\; 
I_n - \T V_2 V_2- \T V_3 V_3>0 }. 
\]
\end{enumerate}
\end{lemma}

Next, we take two points $v_0, z$ as 
\begin{align*}
&v_0=[V_0]\in D(L), \quad z=[Z]\in A\Sigma, 
\\
&V_0=-\sqrt{2}\mat{I_n \\ 0 \\ 0}=AU_0, \quad 
Z=\frac{1}{\sqrt{2}}\mat{-I_n \\ 0 \\ I_n}=AW. 
\end{align*}

For two points $v_1, v_2\in D(L)$ expressed as
\[
v_j=[V_j], \quad V_j=
\mat{V_{j,1} \\ V_{j,2} \\ V_{j,3}}
\quad (j=1,2), 
\]
we define $\psi(v_1,v_2)\in \rM_n(\bbR)$ by
\begin{align*}
\psi(v_1,v_2)&=\T (V_1 \Gamma_1^{-1})L(V_2\Gamma_2^{-1}) , 
\end{align*}
\[
\Gamma_j=\T Z V_j=(V_{j,3}-V_{j,1})/{\sqrt{2}} \quad (j=1,2). 
\]
\begin{lemma}
For any $u_1, u_2\in D(J)$, it holds that
\[
\psi(Au_1, Au_2)=\Delta(u_1,u_2). 
\]
\end{lemma}
By above lemma, we have the following commutative diagram: 
\[
\begin{xy}
\xymatrix{
D(J)\times D(J) \ar[r]^-{A\times A} \ar[dr]^{\circlearrowright}_-{\Delta(\cdot,\cdot)}& D(L)\times D(L) \ar[d]^-{\psi(\cdot,\cdot)}
\\
&\Sym^+_n(\bbR)
}
\end{xy}
\]

\begin{proof}
Take arbitrary $u_1, u_2\in D(J)$ and express them as
\[
u_j=[U_j]
\quad (j=1,2). 
\]
With the relation
\[
\Gamma_j=\T Z (AU_j)=\T WU_j, 
\]
we calculate as
\begin{align*}
\psi(Au_1,Au_2)&=\T (AU_1 \Gamma_1^{-1})L(AU_2 \Gamma_2^{-1})
\\
&=\T (U_1 \Gamma_1^{-1})J(U_2 \Gamma_2^{-1})
\\
&= \Delta(u,u). 
\end{align*}
\end{proof}

Put
$\fkp(L)=\ckakko{X\in \fkg(L);\; X=\T X}$. 

\begin{proposition}
For $B\in \fkp(L)$, assume the following Taylor expansion 
\[
\psi(z, e^{sB}z)=(a\T a) \frac{s^2}{2} +O(s^3), \quad a\in \rM_{n,r}(\bbR)
\]
at $s=0$. 
Define a $\Sym^+_n(\bbR)$-valued function $h(s)$ by
\[
h(s)=\psi(e^{sB}v_0, e^{sB}v_0)^{-1}, 
\]
Then, $h(s)$ satisfies the matrix-valued Bratu equation; 
\[
(h'h^{-1})'=2(a\T a) h^{-1}. 
\]
\end{proposition}

\paragraph{Derivation of the power series solution}
We take an arbitrary $B\in \fkp(L)$ and express it
\[
B=\mat{0&\T C \\ C&0}, \quad C\in \rM_{n+r,n}(\bbR). 
\]
We consider the singular value decomposition of the submatrix $C$ in the following form: 
\[
C=\mat{p_1&p_3 \\ p_2&p_{}}\mat{0 \\ \sigma}\T q, 
\]
\[
\mat{p_1&p_3 \\ p_2&p_{}}\in \mathrm{O}(n+r), \quad 
q\in \mathrm{O}(n), 
\]
\[
\sigma=\diag(\sigma_i)_{i=1}^n \in \rM_n(\bbR). 
\]
We define the following functions of $s\in \bbR$: 
\[
\mathrm{ch}(s\sigma):=\frac12(e^{s\sigma}+e^{-s\sigma}), \quad \mathrm{sh}(s\sigma):=\frac12(e^{s\sigma}-e^{-s\sigma}). 
\]
\begin{proposition}
For $B\in \fkp(L)$, define a function
\[
h(s):=\psi(e^{sB}v_0, e^{sB}v_0)^{-1}. 
\]
Then, $h(s)$ has the following expression: 
\begin{align*}
h(s)&=\frac12\bme(s)\T\bme(s), \quad \bme(s)=p_{} \mathrm{sh}(s\sigma)  - q\mathrm{ch}(s\sigma). 
\end{align*}
\end{proposition}

\begin{proof}
From the singular value decomposition, we get the equation
\[
B=P\mat{0&0&\sigma \\ 0&0&0 \\ \sigma&0&0}\T P, \quad
P:= \mat{q&0&0 \\ 0&p_1 & p_3 \\ 0&p_2&p_{}}. 
\]
Then, we can calculate its exponential as
\begin{align*}
e^{sB}V_0
&=P\mat{\mathrm{ch}(s\sigma)&0&\mathrm{sh}(s\sigma) \\ 0&I_r&0 \\ \mathrm{sh}(s\sigma)&0&\mathrm{ch}(s\sigma)}\T P V_0
\\
&=\sqrt{2}\mat{q\mathrm{ch}(s\sigma) \T q \\ p_3 \mathrm{sh}(s\sigma) \T q \\ p_{} \mathrm{sh}(s\sigma) \T q}. 
\end{align*}
Putting
\[
\Gamma:=p_{} \mathrm{sh}(s\sigma) \T q - q\mathrm{ch}(s\sigma)\T q, \quad
\bme(s)=\Gamma q, 
\]
we obtain that
\begin{align*}
\psi(e^{sB}v_0, e^{sB}v_0)
&=\T (e^{sB}V_0 \Gamma^{-1})L(e^{sB}V_0 \Gamma^{-1}) 
\\
&=2\T \Gamma^{-1}\Gamma^{-1}
=2\T\bme(s)^{-1}\bme(s)^{-1}. 
\end{align*}
\end{proof}

\begin{cor}
$h(s)$ has the following power series expansion: 
\begin{align*}
h(s)&=\frac12 I_{n}
+\frac12\sum_{n = 1}^{\infty} \kakko{q \sigma^{2 n}\; \T q+p_{} \sigma^{2 n}\; \T p_{}} \frac{s^{2 n}}{(2 n) !}
\\
&-\frac12\sum_{n = 1}^{\infty} \kakko{p_{} \sigma^{2 n-1}\;\T q+q\sigma^{2 n-1}\; \T p_{}} \frac{s^{2 n-1}}{(2 n-1) !}. 
\end{align*}
\end{cor}
\begin{proof}
It is shown by the direct calculation: 
\begin{align*}
&\frac12 \bme(s)\T \bme(s)
\\
&=\frac12 \mat{-q & p}\mat{\mathrm{ch}(s\sigma)^2&\mathrm{ch}(s\sigma) \mathrm{sh}(s\sigma) \\ \mathrm{sh}(s\sigma) \mathrm{ch}(s\sigma)&\mathrm{sh}(s\sigma)^2}\mat{-\T q \\ \T p}
\\
&=\frac14 \mat{-q & p}\mat{\mathrm{ch}(2s\sigma)+I_n&\mathrm{sh}(2s\sigma) \\ \mathrm{sh}(2s\sigma) &\mathrm{ch}(2s\sigma)-I_n}\mat{-\T q \\ \T p}
\\
&=\frac14\{2I_n+q(\mathrm{ch}(2s\sigma)-I_n)\T q +p(\mathrm{ch}(2s\sigma)-I_n)\T p 
\\
&\quad -q\mathrm{sh}(2s\sigma)\T p-p \mathrm{sh}(2s\sigma)\T q\}. 
\end{align*}
\end{proof}

\clearpage

\section{Analog to symmetric domain of type CI}

\paragraph{Manifold $\Omega(J)$}
We define a submanifold $\Omega(J)$ in the cone $\Sym^+_{2n}(\bbR)$ by
\[
\Omega(J)=\ckakko{
G\in \Sym^+_{2n}(\bbR);\; JGJ=G^{-1}
}, 
\]
\[
J=\mat{0&iI_n \\  -iI_n&0}. 
\]
Setting the set $V(J)$ by
\[
V(J)=\ckakko{B\in \Sym_{2n}(\bbR);\; JBJ=-B}, 
\]
we can consider the following map: 
\[
V(J)\ni B \mapsto \exp B \in \Omega(J). 
\]

We see that any element $B\in V(J)$ has the following form: 
\[
B=B(b,c):=\mat{b&c  \\ c&-b}, \quad
\begin{array}{l}
b, c\in \Sym_n(\bbR). 
\end{array}
\]
For an element $B\in V(J)$, we define an $\Omega(J)$-valued function $G(s;B)$ by
\[
G(s; B) :=\exp(sB(b,c)) \quad (s\in \bbR). 
\]

We define the following sets of matrices: 
\begin{align*}
&\cG(J)=\ckakko{G\in \GL_{2n}(\bbR);\; JGJ=\T G^{-1}}\cong \mathrm{Sp}(n, \bbR), 
\\
&\cN(J)
=\ckakko{N\in \cG(J);\; N=\mat{I_n&0 \\ *&I_n}}, 
\\
&\qquad=\ckakko{N=\mat{I_n &0  \\ n_1&I_n}, \; n_1= \T n_1}, 
\\
&\Omega_0(J)=\ckakko{A\in \Omega(J);\; A=\mat{*&0 \\ 0&* }}, 
\\
&\qquad=\ckakko{A=\mat{h&0 \\ 0&h^{-1}}, \T h=h}. 
\end{align*}

\begin{proposition}[Block-Gauss decomposition for $\Omega(J)$]
The following map is bijective: 
\[
\begin{array}{ccc}
\cN(J)\times \Omega_0(J) &\rightarrow& \Omega(J)\\
 (N, A)&\mapsto& NA\T N. 
 \end{array}
\]
\end{proposition}

We denote the variable change defined through this decomposition by
\[
G \rightarrow N, A \rightarrow h, n_1, 
\]
or simply by $G  \rightarrow h$.

\begin{lemma}
Under the variable change
\[
G \rightarrow N, A \rightarrow h, n_1, 
\]
it holds that
\begin{align*}
&\frac12\tr((G' G^{-1})^2)
\\
&=\frac12\tr(( A' A^{-1})^2)
+\tr(A\T N' \T N^{-1} A^{-1}N^{-1} N')
\\
&=\tr(( h' h^{-1})^2)
+\tr(h\T n_1' h n_1'). 
\end{align*}
\end{lemma}
%

By the direct derivation of the Euler-Lagrange equation, we obtain the following analog of the matrix-valued Bratu equation.

\begin{theorem}
The equation
\[
G'(s)G(s)^{-1}=B(b,c)
\]
for an $\Omega(J)$-valued function $G(s)$ is expressed as
\[
({h}^{\prime}(s) {h}(s)^{-1})^{\prime}=(c{h}(s)^{-1})^2
\]
by the variable change $G(s)\rightarrow h(s), n_1(s)$. 
\end{theorem}

%

\subsection{Siegel domain $D(J;K)$ of type CI}
Similarly as the type of BDI, 
we put
\begin{align*}
&M=\ckakko{U=\mat{U_1 \\ U_2 }\in \rM_{2n,n}(\bbC);\; \rank U=n}
\end{align*}
and consider the actions
\begin{align*}
&\GL(2n,\bbC) \curvearrowright M \curvearrowleft \GL(n,\bbC)
\end{align*}
by the matrix multiplication. Put the quotient space
\[
D:=M/\GL(n,\bbC)=\ckakko{u=[U]=\bmat{U_1 \\ U_2 };\; U\in M}
\]
and regard it as a domain. We put
\[
J=\mat{0&iI_n \\  -iI_n&0}, \quad K=\mat{0&I_n \\  -I_n&0}
\]
and define a subdomain $D(J;K)$ by
\[
M(J;K)=\ckakko{U\in M;\; U^*JU>0, \T UKU=0}, 
\]
\[
D(J;K):=\ckakko{u=[U]\in D;\; U\in M(J;K)}. 
\]

\begin{lemma}[Realization as a Siegel upper half space]
$D(J;K)$ is the following set: 
\begin{align*}
D(J;K)&= \ckakko{u=\bmat{U_1  \\ I_n};\; \frac{1}{i}(U_1- U_1^*)>0, \T U_1=U_1 }
\\
&\cong \Sp(n,\bbR)/\rU(n). 
\end{align*}
\end{lemma}
We define a subdomain $\Sigma \subset D$ (Shilov boundary of $D(J;K)$) by
\[
\Sigma= \ckakko{u=\bmat{U_1  \\ I_n};\; \frac{1}{i}(U_1- U_1^*)=0, \T U_1=U_1 }. 
\]

We take two points $u_0, w$ as
\begin{align*}
&u_0=[U_0]\in D(J;K), \quad w=[W]\in \Sigma, 
\\
&U_0=\mat{iI_n  \\ I_n}, \quad W=\mat{ 0 \\ I_n}. 
\end{align*}
For two points $u_1, u_2\in D(J;K)$ expressed as
\[
u_j=[U_j], \quad U_j=
\mat{U_{j,1} \\ U_{j,2} }
\quad (j=1,2), 
\]
we define $\Delta(u_1,u_2)\in \rM_n(\bbC)$ by
\begin{align*}
\Delta(u_1,u_2)&= (U_1 \Gamma_1^{-1})^*J(U_2\Gamma_2^{-1}), 
\end{align*}
\[
\Gamma_j=\T W U_j=U_{j,2} \quad (j=1,2). 
\]

Put
\begin{align*}
&\fkg(J)=\ckakko{X\in \gl(2n+r,\bbR);\; XJ +J\T X=0}
\\
&\qquad\cong \fsp(n,\bbR), 
\\
&\fkp(J)=\ckakko{X\in \fkg(J);\; X=\T X}
\\
& =\ckakko{
B\in \fkg(J);\; B=\mat{b&c \\   c&-b}, 
\begin{array}{l}
b, c\in \Sym_n(\bbR)
\end{array}
}. 
\end{align*}

\begin{proposition}[Analog to type CI]\label{prop_CI}
For $B\in \fkp(J)$, define a $\Sym^+_n(\bbR)$-valued function $h(s)$ by
\[
h(s)=\Delta(e^{sB}u_0, e^{sB}u_0)^{-1}. 
\]
Then, $h(s)$ satisfies the equation
\[
(h'h^{-1})'=(ch^{-1})^2. 
\]
\end{proposition}

We can define an action $\cG(J) \curvearrowright D(J:K)$ by
\[
g\cdot u=[gU] \quad (g\in \cG(J), u=[U]\in D(J;K)). 
\]
We also define subgroups: 
\begin{align*}
&\cP(J)=\ckakko{g\in \cG(J);\; g=\mat{*&* \\ 0&* }}, 
\\
&\cK(J)=\ckakko{g\in \cG(J);\; \T g=g^{-1}}\cong\rU(n), 
\\
&\cG(J)_{u_0}=\ckakko{g\in \cG(J);\; gu_0=u_0}. 
\end{align*}

\begin{lemma}
The following assertions hold. 
\begin{enumerate}
\item
$\cG(J)=\cP(J)\cK(J)$. 
\item
$\cG(J)_{u_0}=\cK(J)$. 
\item
The action $\cG\curvearrowright D(J;K)$ is transitive and the following expression holds: 
\[
D(J;K)
\cong \Sp(n,\bbR)/\rU(n). 
\]
\end{enumerate}
\end{lemma}
%

\paragraph{Proof of the proposition}
Define a manifold
\begin{align*}
\Omega(J)&=\ckakko{g\in \Herm^+_{2n}(\bbC);\; gJ g^*=J, gK\T g=K}
\\
&=\ckakko{g\in \Sym^+_{2n}(\bbR);\; gJ\T g=J}
\end{align*}
and a map $\cF: \Omega(J)\xrightarrow{} D$ by
\[
\cF(G)=\bmat{ iI_n-G_2  \\ h}, \quad
G=\mat{h &* \\ G_2&*}\in \Omega(J;K), 
\]
and a map $\pi: \Omega(J)\rightarrow\Sym^+_n(\bbR)$ by
\[
\pi(G)=h, \quad
G=\mat{h &* \\ G_2&* }\in \Omega(J;K). 
\]
We define an action $\cG(J)\curvearrowright \Omega(J)$ by
\[
g\cdot G=\T g^{-1}Gg^{-1}\quad (g\in \cG(J), G\in \Omega(J)). 
\]

\begin{lemma}
The following assertions hold. 
\begin{enumerate}
\item
$\cF(g\cdot G)=g\cF(G)$
for any $g\in \cG(J), G\in \Omega(J)$. 
\item
$\cF(\Omega(J))\subset D(J)$ and the map $\cF: \Omega(J)
\rightarrow D(J)$ is bijective. 
\item
For any $G\in \Omega(J)$, it holds that
\[
\frac12\pi(G)=\Delta(\cF(G),\cF(G))^{-1}. 
\]
\end{enumerate}
\end{lemma}

From the above lemma, we have the following commutative diagram: 
\[
\begin{xy}
\xymatrix{
\Omega(J) \ar[r]^-{\cF} \ar[dr]^{\circlearrowright}_-{\frac12 \pi}& D(J;K) \ar[d]^-{\Delta(\cdot,\cdot)^{-1}}
\\
&\Sym^+_n(\bbR)
}
\end{xy}
\]
\begin{proof}
Each $G\in \Omega(J)$ is expressed as
\[
G=\T g^{-1} g^{-1}
=\mat{
h&hn_1 \\
\T n_1 h&I_n+\T n_1 h n_1
}, 
\]
\[
\T g^{-1}=\mat{I_n&0 \\ \T n_1&I_r}
\mat{h^{\frac12}&0 \\ 0&h^{-\frac12}}\in \cG(J). 
\]
We also see that 
\begin{align*}
gu_0
&=\mat{I&-n_1 \\ 0&I }
\mat{h^{\frac12}&0 \\ 0&h^{-\frac12}}
\bmat{iI_n  \\ I_n}
\\
&
=\bmat{ih^{-1}- n_1   \\ I_n}. 
\end{align*}

{\it 1}. 
First, we show in the case $G=I_{2n}$. 
Take an arbitrary $g\in \cG(J)$ and express it as
\[
g=pk, \quad p\in \cP(J), \; k\in \cK(J), 
\]
\[
p=\mat{p_1&-p_2 \\ 0&\T p_1^{-1}}. 
\]
Then, using the relation $\T p^{-1}=J^{-1}pJ$, we get
\begin{align*}
p\cdot I_{2n}&=\T p^{-1}Gp^{-1}
\\
&=\mat{\T p_1^{-1}&0  \\ p_2& p_1}
\mat{ p_1^{-1}&\T p_2 \\  0& \T p_1}
\\
&=\mat{
\T p_1^{-1}p_1^{-1} &* \\
 p_2 p_1^{-1}&*}. 
\end{align*}
Therefore, 
\begin{align*}
\cF(p\cdot I_{2n})&
=\bmat{iI_n-p_2 p_1^{-1}\\ \T p_1^{-1}p_1^{-1}}
=\bmat{ip_1-p_2 \\ \T p_1^{-1}}. 
\end{align*}
On the other hand, we see that
\begin{align*}
p\cF(I_{2n})=p\bmat{iI_n  \\ I_n}
=\bmat{ip_1-p_2 \\ \T p_1^{-1}}
\end{align*}
and then $\cF(p\cdot I_{2n})=p\cF(I_{2n})$. Therefore, 
\begin{align*}
&\cF(g\cdot I_{2n})=\cF(p\cdot I_{2n})
\\
&=pu_0=pku_0=g\cF(I_{2n}). 
\end{align*}

Next, we show in the general case $g\in \cG(J), G\in \Omega(J)$. 
Express an arbitrary $G\in \Omega(J)$ as
\[
G=g_0\cdot I_{2n} , \quad g_0\in \cG(J). 
\]
Then we see that
\begin{align*}
&\cF(g\cdot G)=\cF(gg_0\cdot I_{2n})
=gg_0\cF( I_{2n})
\\
&=g\cF(g_0 \cdot I_{2n})
=g\cF(G). 
\end{align*}


{\it 3}. 
Take an arbitrary $G\in \Omega(J)$ and express it as
\begin{align*}
G=\mat{h & \T G_2 \\ G_2 & *}, \quad
\T G_2 h-hG_2=0. 
\end{align*}
Then we see that
\[
\cF(G)=[U], \quad U=\mat{iI_n-G_2   \\ h}. 
\]
Putting
\[
\Gamma=\T WU=h, 
\]
we see that
\begin{align*}
&\Delta(\cF(G), \cF(G))= (U\Gamma^{-1})^*J(U\Gamma^{-1})
\\
&=\frac{1}{i}\ckakko{(ih^{-1}-G_2 h^{-1})-(ih^{-1}-G_2 h^{-1})^*}
\\
&=2h^{-1} >0. 
\end{align*}
\end{proof}

\noindent
{\it Proof of Proposition \ref{prop_CI}. }
From the definition, we have
\[
\cF^{-1}(e^{sB}u_0)=e^{-2sB}=:G(s). 
\]
Then, the function $\tilde{h}(s):=\pi(G(s))$ satisfies 
\[
(\tilde{h}'\tilde{h}^{-1})'
=4(ch^{-1})^2. 
\]
Therefore, the function
\[
{h}(s):=\Delta(e^{sB}u_0, e^{sB}u_0)^{-1}=\frac12 \tilde{h}(s)
\]
satisfies
\[
({h}'{h}^{-1})'=(ch^{-1})^2. 
\]
\qed

%
%
%
%

\section{Remarks}
We gave individual calculations for each symmetric domain of type BDI and CI. The same calculation, which consists of the Lagrangian calculation and the realization of symmetric domains, is likely applicable to other symmetric domains classified by Cartan, ref \cite{H,PS}. For this purpose, it is natural to expect a unified way to treat general symmetric domains.

\vspace{1cm}
\begin{flushright}
{\it
\begin{tabular}{l}
Hiroto Inoue
\\
Institute of Mathematics for Industry
\\
Kyushu University
\\
744 Motooka, Nishi-ku
\\
Fukuoka, 819-0395, Japan
\\
(E-mail: hi-inoue@math.kyushu-u.ac.jp)
\end{tabular}
}
\end{flushright}

\end{document}